%% file: maineptcs.tex
\documentclass[creativecommons]{eptcs}

\providecommand{\event}{GandAlf 2025} 

\usepackage{iftex}

\ifpdf
  \usepackage{underscore}         
  \usepackage[T1]{fontenc}        
\else
  \usepackage{breakurl}           
\fi

\input{preamble}

\newtheorem{theorem}{Theorem}
\newtheorem{corollary}{Corollary}
\newtheorem{remark}{Remark}
\newtheorem{lemma}{Lemma}

\title{The Complexity of Generalized HyperLTL\\ with Stuttering and Contexts}

\author{Gaëtan Regaud\thanks{Supported by the European Union.\newline
\includegraphics[scale=.10]{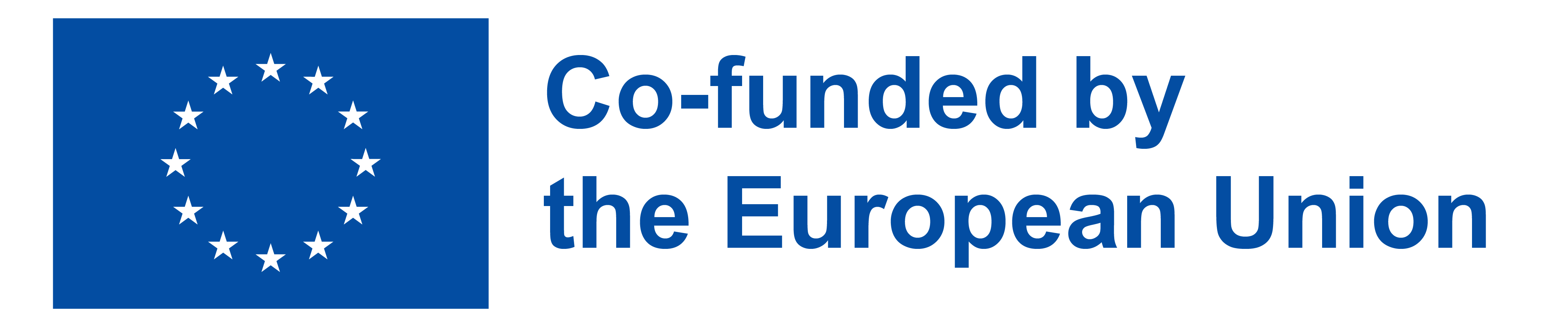}}
\institute{ENS Rennes\\ Rennes, France}
\email{gaetan.regaud@ens-rennes.fr}
\and
Martin Zimmermann\thanks{Supported by DIREC - Digital Research Centre Denmark.}
\institute{Aalborg University\\ Aalborg, Denmark}
\email{mzi@cs.aau.dk}
}
\def\titlerunning{The Complexity of Generalized HyperLTL with Stuttering and Contexts}

\def\authorrunning{G.\ Regaud and M.\ Zimmermann}

\begin{document}
\maketitle

\begin{abstract}
\input{abstract}
\end{abstract}

\input{content}

\bibliographystyle{eptcs}
\bibliography{bib}


\end{document}

%% file: preamble.tex
\usepackage[utf8]{inputenc}

\usepackage{amsmath}
\usepackage{amssymb}

\usepackage{amsthm}
\usepackage{tablefootnote}

\usepackage{xspace}

\usepackage{colortbl}

\usepackage{booktabs}

\usepackage{mathtools}

\usepackage{tikz}
\usetikzlibrary{arrows,decorations.text,decorations.pathmorphing,decorations.pathreplacing,positioning,patterns,tikzmark,arrows.meta,backgrounds,shapes.geometric}

\newcommand{\nats}{\mathbb{N}}
\renewcommand{\epsilon}{\varepsilon}
\renewcommand{\phi}{\varphi}

\newcommand{\pow}[1]{2^{#1}}
\newcommand{\cceq}{\mathop{::=}}
\newcommand{\set}[1]{\{#1\}}

\newcommand{\F}{\mathop{\mathbf{F}\vphantom{a}}\nolimits}
\newcommand{\G}{\mathop{\mathbf{G}\vphantom{a}}\nolimits}
\DeclareMathOperator{\U}{\mathbf{U}}
\newcommand{\X}{\mathop{\mathbf{X}\vphantom{a}}\nolimits}

\renewcommand{\O}{\mathop{\mathbf{O}\vphantom{a}}\nolimits}
\renewcommand{\H}{\mathop{\mathbf{H}\vphantom{a}}\nolimits}
\DeclareMathOperator{\Since}{\mathbf{S}}
\newcommand{\Y}{\mathop{\mathbf{Y}\vphantom{a}}\nolimits}

\newcommand{\ltl}{{LTL}\xspace}
\newcommand{\pdl}{{PDL}\xspace}
\newcommand{\ctl}{{CTL}\xspace}
\newcommand{\ctlstar}{{CTL$^*$}\xspace}
\newcommand{\hyltl}{{Hyper\-LTL}\xspace}
\newcommand{\hypdl}{{Hyper\-PDL-$\Delta$}\xspace}
\newcommand{\hyqptl}{{Hyper\-QPTL}\xspace}
\newcommand{\sohyltl}{{Hyper$^2$LTL}\xspace}
\newcommand{\hmu}{{$H_\mu$}\xspace}

\newcommand{\hyctlstar}{{HyperCTL$^*$}\xspace}

\newcommand{\qptl}{{QPTL}\xspace}
\newcommand{\hyqptlplus}{{Hyper\-QPTL$^+$}\xspace}
\newcommand{\ahyltl}{{A-HLTL}\xspace}
\newcommand{\hyaut}{HA\xspace}
\newcommand{\foe}{FO$[E,<]$\xspace}
\newcommand{\hyperfo}{HyperFO\xspace}
\newcommand{\sonese}{S1S$[E,<]$\xspace}

\newcommand{\ap}[0]{\mathrm{AP}}

\newcommand{\tower}{\textsc{Tower}\xspace}

\newcommand{\myquot}[1]{``#1''}

\newcommand{\traces}{\mathrm{Tr}}

\newcommand{\hyperize}{{\mathit{hyp}}}

\newcommand{\natsstruct}{(\nats, +, \cdot, <, \in)}

\newcommand{\ghltl}{\textrm{GHyLTL$_\textrm{S+C}$}\xspace}

\newcommand{\hltls}{\textrm{HyperLTL$_\textrm S$}\xspace}
\newcommand{\hltlc}{\textrm{HyperLTL$_\textrm C$}\xspace}
\renewcommand{\phi}{\varphi}
\newcommand{\vars}{\textsf{VAR}}
\newcommand{\tra}{\sigma}
\newcommand{\asg}{\Pi}

\newcommand{\p}{p}

\newcommand{\tr}{x}
\newcommand{\C}[1]{\langle #1 \rangle}
\newcommand{\ctx}{C}

\newcommand{\intdeco}{\texttt{pos}}

\newcommand{\setS}{\Gamma\xspace}
\newcommand{\setL}{\mathcal{L}\xspace}

\newcommand{\pltl}{{PLTL}\xspace}

\newcommand{\dom}[1]{\textrm{Dom}(#1)\xspace}

\newcommand{\TS}{\mathcal{T}\xspace}

\newcommand{\labfunc}{\ell}

\renewcommand{\succ}{\mathrm{succ}}
\newcommand{\pred}{\mathrm{pred}}

\newcommand{\inprop}{\#}
\newcommand{\inpropm}[1]{\#(#1)}
\newcommand{\auxprop}{\$}
\newcommand{\auxpropp}{\$'}
\newcommand{\traux}{x}
\newcommand{\trauxp}{x'}
\newcommand{\adddeco}{\mathrm{add}}
\newcommand{\multdeco}{\mathrm{mult}}

\newcommand{\per}{\mathrm{per}}
\newcommand{\blockchange}{\mathrm{algn}}

\newcommand{\true}{\texttt{true}}
\newcommand{\intprop}{\texttt{\#}}

%% file: abstract.tex
We settle the complexity of satisfiability and model-checking for generalized HyperLTL with stuttering and contexts, an expressive logic for the specification of asynchronous hyperproperties.
Such properties cannot be specified in HyperLTL, as it is restricted to synchronous hyperproperties.

Nevertheless, we prove that satisfiability is $\Sigma_1^1$-complete and thus not harder than for HyperLTL. 
On the other hand, we prove that model-checking is equivalent to truth in second-order arithmetic, and thus much harder than the decidable HyperLTL model-checking problem. 
The lower bounds for the model-checking problem hold even when only allowing stuttering or only allowing contexts.

%% file: content.tex
\section{Introduction}

The introduction of hyperlogics has been an important milestone in the specification, analysis, and verification of hyperproperties~\cite{ClarksonS10}, properties that relate several execution traces of a system. 
These have important applications in, e.g., information-flow security. 
Before their introduction, temporal logics (e.g., \ltl, \ctl, \ctlstar, \qptl and \pdl) were only able to reason about a single trace at a time~\cite{DBLP:journals/eatcs/Finkbeiner17}.
However, this is not sufficient to reason about the complex flow of information.
For example, noninterference~\cite{noninterference} requires that all traces that coincide on their low-security inputs also coincide on their low-security outputs, independently of their high-security inputs (which may differ, but may not leak via low-security~outputs).
This property intuitively \myquot{compares} the inputs and outputs on pairs of traces and is therefore not expressible in LTL, which can only reason about individual traces.

The first generation of hyperlogics has been introduced by equipping \ltl, \ctlstar, \qptl, and \pdl with quantification over traces, obtaining \hyltl~\cite{ClarksonFKMRS14}, \hyctlstar~\cite{ClarksonFKMRS14}, \hyqptl~\cite{FinkbeinerHHT20,Rabe16diss} and \hypdl~\cite{hyperpdl}. They are able to express noninterference (and many other hyperproperties), have intuitive syntax and semantics, and a decidable model-checking problem, making them attractive specification languages for hyperproperties.
For example, noninterference is expressed in \hyltl as
\[
\forall \pi.\ \forall \pi'.\ \left(\bigwedge\nolimits_{i \in I_\ell} \G \left(i_\pi \leftrightarrow i_{\pi'} \right)\right) \rightarrow \left(\bigwedge\nolimits_{o \in O_\ell} \G \left(o_\pi \leftrightarrow o_{\pi'} \right)\right),
\]
where $I_\ell$ is the set of low-security inputs and $O_\ell$ is the set of low-security outputs. 
All these logics are synchronous in the sense that time passes on all quantified traces at the same rate. 

However, not every system is synchronous, e.g., multi-threaded systems in which processes are not scheduled in lockstep. 
The first generation of hyperlogics is not able to express asynchronous hyperproperties. 
Hence, in a second wave, several asynchronous hyperlogics have been introduced, which employ various mechanisms to enable the asynchronous evolution of time on different traces under consideration.
\begin{itemize}
    \item Asynchronous \hyltl (\ahyltl)~\cite{DBLP:conf/cav/BaumeisterCBFS21} adds so-called trajectories to \hyltl, which intuitively specify the rates at which different traces evolve. 
    \item \hyltl with stuttering (\hltls)~\cite{DBLP:conf/lics/BozzelliPS21} changes the semantics of the temporal operators of \hyltl so that time does not evolve synchronously on all traces, but instead evolves based on \ltl-definable stuttering.
    \item \hyltl with contexts (\hltlc)~\cite{DBLP:conf/lics/BozzelliPS21} adds a context-operator to \hyltl, which allows to select a subset of traces on which time passes synchronously, while it is frozen on all others. 
    \item Generalized \hyltl with stuttering and contexts (\ghltl)~\cite{DBLP:conf/fsttcs/BombardelliB0T24} adds both stuttering and contexts to \hyltl and additionally allows trace quantification under the scope of temporal operators, which \hyltl does not allow. 
    \item \hmu~\cite{hmu} adds trace quantification to the linear-time $\mu$-calculus with asynchronous semantics for the modal operators. 
    \item Hypernode automata (\hyaut)~\cite{hypernode} combine automata and hyperlogic with stuttering. 
    \item First- and second-order predicate logic with the equal-level predicate (\foe, \hyperfo, and \sonese)~\cite{FZ17,CoenenFHH19} (evaluated over sets of traces) can also be seen as asynchronous hyperlogics. 
\end{itemize}
The known relations between these logics are depicted in Figure~\ref{figasynchlogics}.

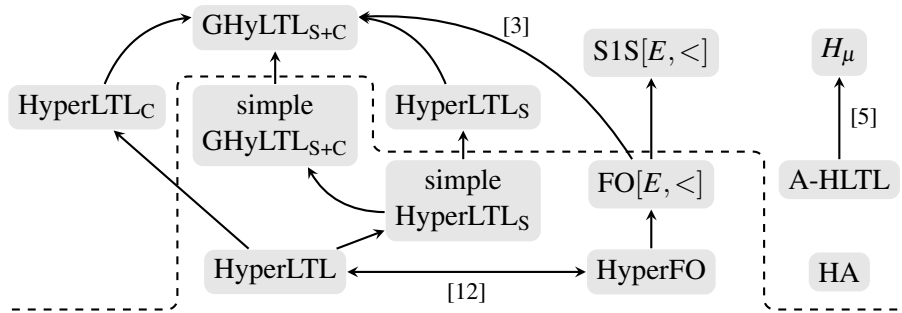
\begin{figure}[b]
    \centering

    \begin{tikzpicture}[thick]
        \node[fill=gray!20,rounded corners] (hyltl) at (0,0.5) {\hyltl};
        \node[fill=gray!20,rounded corners,align=center,anchor = north] (shltls) at (2.5,2) {simple\\ \hltls};
        \node[fill=gray!20,rounded corners,align=center,anchor = north] (hltls) at (2.5,3) {\hltls};
        \node[fill=gray!20,rounded corners,align=center,anchor = north] (hltlc) at (-2.5,3) {\hltlc};
        \node[fill=gray!20,rounded corners,align=center,anchor = north] (sghltl) at (0,3) {simple\\ \ghltl};
        \node[fill=gray!20,rounded corners,align=center] (ghltl) at (0,3.75) {\ghltl};

        \node[fill=gray!20,rounded corners,align=center,anchor = north] (sonese) at (5,3.75) {\sonese};
        \node[fill=gray!20,rounded corners,align=center,anchor = north] (foe) at (5,2) {\foe};
        \node[fill=gray!20,rounded corners,align=center] (hyperfo) at (5,0.5) {\hyperfo};
        
        \node[fill=gray!20,rounded corners,align=center] (hyaut) at (7.5,0.5) {\hyaut};
        \node[fill=gray!20,rounded corners,align=center,anchor = north] (ahyltl) at (7.5,2) {\ahyltl};
        \node[fill=gray!20,rounded corners,align=center,anchor = north] (hmu) at (7.5,3.75) {\hmu};

        \path[->, > = stealth]
        (hyltl) edge node[above] {} (hltlc)
        (hyltl) edge node[above] {} (shltls)
        (shltls) edge node[above] {} (hltls)
        (shltls) edge[bend left] node[above] {} (sghltl) 
        (sghltl) edge node[above] {} (ghltl)
        (hltls) edge[bend right] node[above] {} (ghltl)
        (hltlc) edge[bend left] node[above] {} (ghltl)
        (hyltl) edge[<->] node[below] {\footnotesize \cite{FZ17}} (hyperfo)
        (hyperfo) edge node[above] {} (foe)
        (foe) edge node[above] {} (sonese)
        (ahyltl) edge node[right]{\footnotesize \cite{expressiveness}} (hmu)
        (foe) edge[bend right] node[above] {\footnotesize \cite{DBLP:conf/fsttcs/BombardelliB0T24}} (ghltl)
        ;

        \draw[dashed, rounded corners] (-3.5,0) -- (-1.3,0) -- (-1.3,3.1) -- (1.3, 3.1) -- (1.3, 2.1) -- (6.5,2.1) -- (6.5, 1.1) -- (6.5,0) -- (8.5,0);
        
    \end{tikzpicture}
    
    \caption{The landscape of logics for asynchronous hyperproperties. Arrows denote known inclusions and the dashed line denotes the decidability border for model-checking. For non-inclusions, we refer the reader to work by Bozelli et al.~\cite{expressiveness,DBLP:conf/fsttcs/BombardelliB0T24}.}
    \label{figasynchlogics}
\end{figure}

However, all these logics have an undecidable model-checking problem, thereby losing one of the key features of the first generation logics. 
Thus, much research focus has been put on fragments of these logics, e.g., simple \ghltl and simple \hltls, which both have a decidable model-checking problem. The same is true for fragments of \ahyltl~\cite{DBLP:conf/cav/BaumeisterCBFS21}, \hmu~\cite{hmu}, and \hyaut~\cite{hypernode}.
Furthermore, for almost all of the logics, the satisfiability problem has never been studied. 
Thus, the landscape of complexity results for the second generation is still incomplete, while the complexity of satisfiability and model-checking for the first generation has been settled (see Table~\ref{tablefirstgenresults}).

\begin{table}
    \setlength{\tabcolsep}{15pt}
    \centering
    \caption{List of complexity results for synchronous hyperlogics. \myquot{T3A-equivalent} stands for \myquot{equivalent to truth in third-order arithmetic}. The result for \hypdl satisfiability can be shown using techniques developed  by Fortin et al.\ for \hyltl satisfiability~\cite{hyperltlsat}.} 
    \renewcommand{\arraystretch}{1}
    \begin{tabular}{lll}
    
       Logic &  Satisfiability    &  Model-checking \\
         \midrule
    
        \hyltl & $\Sigma_1^1$-complete~\cite{hyperltlsat}  & \tower-complete~\cite{Rabe16diss,MZ20} \\
        
        \rowcolor{lightgray!40} \hypdl & $\Sigma_1^1$-complete & $\tower$-complete~\cite{hyperpdl}\\
        
        \hyqptl &  $\Sigma_1^2$-complete~\cite{hyq} & $\tower$-complete~\cite{Rabe16diss}\\
        
        \rowcolor{lightgray!40}\hyqptlplus & T3A-equivalent~\cite{hyq} &   T3A-equivalent~\cite{hyq}\\
        
        \sohyltl & T3A-equivalent~\cite{fz25}  &  T3A-equivalent~\cite{fz25}\\
        
        \rowcolor{lightgray!40}\hyctlstar & $\Sigma_1^2$-complete~\cite{hyperltlsat}  & \tower-complete~\cite{Rabe16diss,MZ20} \\
    
    \end{tabular}
    \label{tablefirstgenresults}
\end{table}

In these preceding works, and here, one uses the complexity of arithmetic, predicate logic over the signature~$(+, \cdot, <)$, as a yardstick.
In first-order arithmetic, quantification ranges over natural numbers while second-order arithmetic adds quantification over sets of natural numbers and third-order arithmetic adds quantification over sets of sets of natural numbers.
Figure~\ref{fighierarchies} gives an overview of the arithmetic, analytic, and \myquot{third} hierarchy, each spanned by the classes of languages definable by restricting the number of alternations of the highest-order quantifiers, i.e., $\Sigma_n^0$ contains languages definable by formulas of first-order arithmetic with $n-1$ quantifier alternations, starting with an existential one.

\begin{figure}[b]
    \centering
    
  \scalebox{.9}{
  \begin{tikzpicture}[xscale=1.0,yscale=.7,thick]

    \fill[fill = gray!25, rounded corners] (-1.05,2) rectangle (0.5,-1.5);
    \fill[fill = gray!25, rounded corners] (.6,2) rectangle (14.5,-1.5);

    \node[align=center] (s00) at (0,0) {$\Sigma^0_0$ \\ $=$ \\ $\Pi^0_0$} ;
    \node (s01) at (1,1) {$\Sigma^0_1$} ;
    \node (p01) at (1,-1) {$\Pi^0_1$} ;
    \node (s02) at (2,1) {$\Sigma^0_2$} ;
    \node (p02) at (2,-1) {$\Pi^0_2$} ;
    \node (s03) at (3,1) {$\Sigma^0_3$} ;
    \node (p03) at (3,-1) {$\Pi^0_3$} ;
    \node (s04) at (4,1) {$\cdots$} ;
    \node (p04) at (4,-1) {$\cdots$} ;
    
    \node[align=center] (s10) at (5,0) {$\Sigma^1_0 $ \\ $=$ \\ $ \Pi^1_0$} ;
    \node (s11) at (6,1) {$\Sigma^1_1$} ;
    \node (p11) at (6,-1) {$\Pi^1_1$} ;
    \node (s12) at (7,1) {$\Sigma^1_2$} ;
    \node (p12) at (7,-1) {$\Pi^1_2$} ;
    \node (s13) at (8,1) {$\Sigma^1_3$} ;
    \node (p13) at (8,-1) {$\Pi^1_3$} ;
    \node (s14) at (9,1) {$\cdots$} ;
    \node (p14) at (9,-1) {$\cdots$} ;
    
    \node[align=center] (s20) at (10,0) {$\Sigma^2_0$ \\ $ =$ \\ $ \Pi^2_0$} ;
    \node (s21) at (11,1) {$\Sigma^2_1$} ;
    \node (p21) at (11,-1) {$\Pi^2_1$} ;
    \node (s22) at (12,1) {$\Sigma^2_2$} ;
    \node (p22) at (12,-1) {$\Pi^2_2$} ;
    \node (s23) at (13,1) {$\Sigma^2_3$} ;
    \node (p23) at (13,-1) {$\Pi^2_3$} ;
    \node (s24) at (14,1) {$\cdots$} ;
    \node (p24) at (14,-1) {$\cdots$} ;
    
    \foreach \i in {0,1,2} {
      \draw (s\i0) -- (s\i1) ;
      \draw (s\i0) -- (p\i1) ;
      \draw (s\i1) -- (s\i2) ;
      \draw (s\i1) -- (p\i2) ;
      \draw (p\i1) -- (s\i2) ;
      \draw (p\i1) -- (p\i2) ;
      \draw (s\i2) -- (s\i3) ;
      \draw (s\i2) -- (p\i3) ;
      \draw (p\i2) -- (s\i3) ;
      \draw (p\i2) -- (p\i3) ;
      \draw (s\i3) -- (s\i4) ;
      \draw (s\i3) -- (p\i4) ;
      \draw (p\i3) -- (s\i4) ;
      \draw (p\i3) -- (p\i4) ;
    }
    \foreach \i [evaluate=\i as \iplus using int(\i+1)] in {0,1} {
      \draw (s\i4) -- (s\iplus0) ;
      \draw (p\i4) -- (s\iplus0) ;
}
      \node[] at (-0.25,1.7) {\scriptsize Decidable} ;
      \node at (13.4,1.7) {\scriptsize Undecidable} ;


      \node[align=left,font=\small,
      rounded corners=2pt] at (3.2,1.7)
      (re) {\scriptsize Recursively enumerable} ;
      \draw[-stealth,rounded corners=2pt] (s01) |- (re) ;

    \path (4.25,-1.65) edge[decorate,decoration={brace,amplitude=3pt}]
    node[below] {\begin{minipage}{4cm}\centering
    	arithmetical hierarchy $ $\\ $\equiv$\\ first-order arithmetic 
    \end{minipage}} (.75,-1.65) ;

    \path (9.25,-1.65) edge[decorate,decoration={brace,amplitude=3pt}]
    node[below] {\begin{minipage}{4.4cm}\centering
    	analytical hierarchy $ $\\$\equiv$\\ second-order arithmetic 
    \end{minipage}} (5,-1.65) ;

    \path (14.25,-1.65) edge[decorate,decoration={brace,amplitude=3pt}]
    node[below] {\begin{minipage}{4cm}\centering
    	\myquot{the third hierarchy} $ $\\ {$\equiv$}\\ third-order arithmetic 
    \end{minipage}} (10,-1.65) ;

  \end{tikzpicture}}
    
    \caption{The arithmetical hierarchy, the analytical hierarchy, and beyond.}
    \label{fighierarchies}
\end{figure}
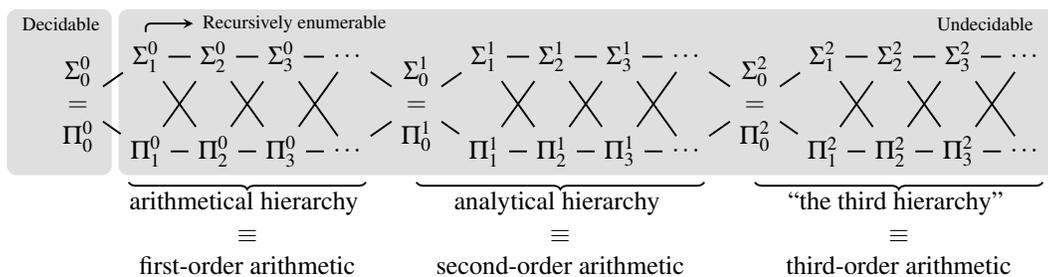

Our goal is to obtain a similarly clear picture for asynchronous logics, both for model-checking (for which, as mentioned above, only some lower bounds are known) and for satisfiability (for which almost nothing is known).
In this work, we focus on \ghltl, as it is one of the most expressive logics and subsumes many of the other logics.

First, we study the satisfiability problem. 
It is known that \hyltl satisfiability is $\Sigma_1^1$-complete.
Here, we show that satisfiability for \ghltl is not harder, i.e., also $\Sigma_1^1$-complete.
The lower bound is trivial, as \hyltl is a fragment of \ghltl. 
However, we show that adding stuttering, contexts, and quantification under the scope of temporal operators all do not increase the complexity of satisfiability.
Intuitively, the underlying reason is that \ghltl is a first-order linear-time logic, i.e., it is evaluated over a set of traces and Skolem functions for the existentially quantified variables map tuples of traces to traces.
We exploit this property to show that every satisfiable formula has a countable model. 
The existence of such a \myquot{small} model can be captured in $\Sigma_1^1$.
This should be contrasted with \hyctlstar, which only adds quantification under the scope of temporal operators to \hyltl, but with a branching-time semantics. 
In \hyctlstar, one can write formulas that have only uncountable models, which in turn allows one to encode existential third-order quantification~\cite{hyperltlsat}.
Consequently, \hyctlstar satisfiability is $\Sigma_1^2$-hard (and in fact $\Sigma_1^2$-complete) and thus much harder than that of \ghltl.
Let us also mention that these results settle the complexity of \foe satisfiability: it is $\Sigma_1^1$-complete as well. Here, the lower bound is inherited from \hyltl and the upper bound follows from the fact that \foe can be translated into \ghltl.

Then, we turn our attention to the model-checking problem, which we show to be equivalent to truth in second-order arithmetic and therefore much harder than satisfiability. 
Here, we show that, surprisingly, the lower bounds already hold for the fragments~\hltls and \hltlc, i.e., adding one feature is sufficient, and adding the second does not increase the complexity further. 
This result also has to be contrasted with \hyltl model-checking: adding stuttering or contexts takes the model-checking problem from \tower-complete~\cite{FinkbeinerRS15} (and thus decidable) to truth in second-order arithmetic. 

The intuitive reason for model-checking being much harder than satisfiability is that every satisfiable formula of \ghltl has a countable model while in the model-checking problem, one has to deal with possibly uncountable models, as (finite) transition systems may have uncountably many traces.
This allows us to encode second-order arithmetic in \hltls  and in \hltlc. 
A similar situation occurs for a fragment of second-order \hyltl~\cite{fz25}, but in general satisfiability is harder than or as hard as model-checking (see Table~\ref{tablefirstgenresults}).

All proofs omitted due to space restrictions can be found in the full version~\cite{fullversion}.

\section{Preliminaries}

The set of nonnegative integers is denoted by $\nats$. 
An alphabet is a nonempty finite set~$\Sigma$.
The set of infinite words over $\Sigma$ is $\Sigma^\omega$. Given $w \in \Sigma^\omega$ and $i \in\nats$, $w(i)$ denotes the $i$-th letter of $w$ (starting with $i =0$).
Let $\ap$ be a fixed finite set of propositions. A trace~$\tra$ is an element of $(\pow\ap)^\omega$, and a pointed trace is a pair $(\tra,i)$ consisting of a trace and a pointer~$i \in\nats$ pointing to a position of $\tra$.
It is initial if $i=0$.

A transition system is a tuple~$\TS=(V,E,I,\labfunc)$ where $V$ is a nonempty finite set of vertices, $E\subseteq V\times V$ is a set of directed edges, $I\subseteq V$ is a set of initial vertices, and $\labfunc: V\to \pow\ap$ is a labeling function that maps each vertex to a set of propositions. We require that each vertex has at least one outgoing edge.
A run of a transition system~$\TS$ is an infinite word~$v_0v_1\cdots\in V^\omega$ such that $v_0\in I$ and $(v_i,v_{i+1})\in E$ for all $i\in\nats$. The set~$\traces(\TS)=\{\labfunc(v_0)\labfunc(v_1)\cdots\mid v_0v_1\cdots\text{is a run of $\TS$}\}$ is the set of traces induced by $\TS$.

\textbf{LTL with Past.}
The logic~\pltl~\cite{Pnueli77} extends classical \ltl~\cite{Pnueli77} by adding temporal operators to describe past events. The syntax of \pltl is defined as 
\[\theta \cceq \p \mid \lnot \theta \mid \theta \lor \theta \mid \X \theta \mid \theta \U \theta \mid \Y \theta \mid \theta \Since \theta,\]
where $\p\in\ap$. Here, $\Y$ (yesterday) and $\Since$ (since) are the past-variants of $\X$ (next) and $\U$ (until).
We use the usual syntactic sugar, e.g., $\land$, $\rightarrow$, $\leftrightarrow$, $\F$ (eventually), $\G$ (always), $\O$ (once, the past-variant of eventually), and $\H$ (historically, the past-variant of always).

The semantics of \pltl is defined over pointed traces $(\tra,i)$ as 
\begin{itemize}
    \item $(\tra,i)\models \p $ if $ \p\in\tra(i)$,
    \item $(\tra,i)\models \lnot\theta $ if $ (\tra,i)\not\models\theta$,
    \item $(\tra,i)\models \theta_1\lor\theta_2 $ if $ (\tra,i)\models\theta_1 $ or $ (\tra,i)\models\theta_2$,
    \item $(\tra,i)\models \X\theta $ if $ (\tra,i+1)\models\theta$,
    \item $(\tra,i)\models \theta_1\U\theta_2 $ if there exists an $ i'\geq i $ such that $(\tra,i')\models\theta_2 $ and $(\tra,j)\models\theta_1$ for all $ i\le j < i'$,
    \item $(\tra,i)\models \Y\theta $ if $ i>0 $ and $(\tra,i-1)\models\theta$, and
    \item $(\tra,i)\models \theta_1\Since\theta_2 $ if there exits an $0 \le i'\leq i $ such that $(\tra,i')\models\theta_2 $ and $(\tra,j)\models\theta_1$ for all $ i'< j\leq i$.
\end{itemize}

\textbf{Stuttering.}
Let $\setS$ be a finite set of \pltl formulas and $\tra$ a trace. 
We say that $i \in \nats$ is a proper $\setS$-changepoint of $\tra$ if either $i = 0$ or  $i >0$ and there is a $\theta \in \setS$ such that $(\tra,i) \models \theta$ if and only if $(\tra,i-1)\not\models\theta$, i.e., the truth value of $\theta$ at positions $i$ and $i-1$ differs.    
If $\tra$ has only finitely many proper $\setS$-changepoints (say $i$ is the largest one), then $i+1, i+2, \ldots$ are $\setS$-changepoints of $\tra$ by convention.
Thus, every trace has infinitely many $\setS$-changepoints.

The $\setS$-successor of a pointed trace~$(\tra,i)$ is the pointed trace~$\succ_\setS(\tra, i) = (\tra,i')$ where $i'$ is the minimal $\setS$-changepoint of $\tra$ that is strictly greater than $i$.
Dually, the $\setS$-predecessor of $(\tra,i)$ for $i >0$ is the pointed trace~$\pred_\setS(\tra,i) = (\tra,i')$ where $i'$ is the maximal $\setS$-changepoint of $\tra$ that is strictly smaller than $i$; $\pred_\setS(\tra,0)$ is undefined.

\begin{remark}
\label{remstutteringspecialcases}
Let $\tra$ be a trace over some set~$\ap'$ of propositions and let $\setS$ only contain \pltl formulas using propositions in $\ap''$ such that $\ap'\cap \ap'' = \emptyset$ (note that this is in particular satisfied, if $\setS=\emptyset$).
Then, $0$ is the only proper $\setS$-changepoint of $\tra$.
Hence, by our convention, every position of $\tra$ is a $\setS$-changepoint which implies $\succ_\setS(\tra, i) = (\tra,i+1)$ for all $i$ and $\pred_\setS(\tra, i) = (\tra,i-1)$ for all $i>0$.
\end{remark}

\textbf{Generalized HyperLTL with Stuttering and Contexts.}
Recall that \hyltl extends \ltl with trace quantification in prenex normal form, i.e., first some traces are quantified and then an \ltl formula is evaluated (synchronously) over these traces.
\ghltl extends \hyltl by two new constructs to express asynchronous hyperproperties:
\begin{itemize}
    \item Contexts allow one to restrict the set of quantified traces over which time passes when evaluating a formula, e.g., $\C \ctx \psi$ for a nonempty finite set~$\ctx$ of trace variables expresses that $\psi$ holds when time passes synchronously on the traces bound to variables in $\ctx$, but time does not pass on variables bound to variables that are not in $\ctx$.
    \item Furthermore, temporal operators are labeled by sets~$\setS$ of \pltl formulas and time stutters w.r.t.\ $\setS$, e.g., $\X_\setS$ stutters to the $\setS$-successor on each trace in the current context.
\end{itemize}
Also, unlike \hyltl, \ghltl allows trace quantification under the scope of temporal operators. 

Fix a finite set~$\vars$ of trace variables. The syntax of \ghltl is given by the grammar
\[\phi\cceq \p_\tr \mid \lnot \phi\mid\phi\lor\phi\mid\C\ctx\phi\mid\X_\setS\phi\mid\phi\U_\setS\phi\mid\Y_\setS\phi\mid\phi\Since_\setS\phi\mid\exists\tr.\phi\mid\forall\tr.\phi,\]
where $\p\in\ap$, $\tr\in\vars$, $\ctx \in\pow\vars\setminus\set{\emptyset}$, and $\setS$ ranges over finite sets of \pltl formulas.
A sentence is a formula without free trace variables, which are defined as expected.
To declutter our notation, we write $\C{\tr_1, \ldots, \tr_n}$ for contexts instead of $\C{\set{\tr_1, \ldots, \tr_n}}$.

To define the semantics of \ghltl we need to introduce some notation.
A (pointed) trace assignment~$\asg\colon\vars\to (\pow\ap)^\omega\times\nats$ is a partial function that maps trace variables to pointed traces.
The domain of a trace assignment~$\asg$, written as $\dom{\asg}$, is the set of variables for which $\asg$ is defined.
For $\tr\in\vars$, $\tra\in(\pow\ap)^\omega$, and $i\in\nats$, the assignment~$\asg[\tr\mapsto(\tra,i)]$ maps $\tr$ to $(\tra,i)$ and each other~$\tr' \in\dom{\asg} \setminus\set{\tr}$ to $\asg(\tr')$.

Fix a set~$\setS$ of \pltl formulas and a context~$\ctx \subseteq \vars$. 
The $(\setS, \ctx)$-successor and the $(\setS, \ctx)$-pre\-de\-ces\-sor of a trace assignment~$\asg$ are the
trace assignments~$\succ_{(\setS, \ctx)}(\asg)$ and $\pred_{(\setS, \ctx)}(\asg(\tr))$ defined as
\[
\succ_{(\setS, \ctx)}(\asg)(\tr) = \begin{cases}
    \succ_\setS(\asg(\tr)) & \text{if $\tr \in \ctx$},\\
    \asg(\tr) &\text{otherwise,}
\end{cases}
\text{ and }
\pred_{(\setS, \ctx)}(\asg)(\tr) = \begin{cases}
    \pred_\setS(\asg(\tr)) & \text{if $\tr \in \ctx$},\\
    \asg(\tr) &\text{otherwise.}
\end{cases}
\]
The $(\setS, \ctx)$-predecessor of $\asg$ is only defined when $\pred_\setS(\asg(\tr))$ is defined for all $\tr \in \ctx$.\footnote{Note that this definition differs from the one in the original paper introducing \ghltl~\cite{DBLP:conf/fsttcs/BombardelliB0T24}, which required that $\pred_\setS(\asg(\tr))$ is defined for every $\tr \in \dom{\asg}$. However, this is too restrictive, as the predecessor operation is only applied to traces~$\asg(\tr)$ with $\tr \in \ctx$. Furthermore, it leads to undesirable side effects, e.g., $\setL \models \phi$ may hold, but $\setL \models \forall \tr. \phi$ does not hold, where $\tr$ is a variable not occurring in $\phi$.}

Iterated $(\setS, \ctx)$-successors and $(\setS, \ctx)$-predecessors are defined as
\begin{itemize}
    \item $\succ^0_{(\setS, \ctx)}(\asg) = \asg$ and $\succ^{j+1}_{(\setS, \ctx)}(\asg) = \succ_{(\setS, \ctx)}(\succ^j_{(\setS, \ctx)}(\asg) )$ and
    \item $\pred^0_{(\setS, \ctx)}(\asg) = \asg$ and $\pred^{j+1}_{(\setS, \ctx)}(\asg) = \pred_{(\setS, \ctx)}(\pred^j_{(\setS, \ctx)}(\asg) )$, which may again be undefined.
\end{itemize}

Now, the semantics of \ghltl is defined with respect to a set~$\setL$  of traces, an assignment~$\asg$, and a context~$\ctx \subseteq \vars$ as
\begin{itemize}

    \item $(\setL, \asg,\ctx)\models\p_\tr$ if $\asg(\tr)=(\tra,i)$ and $ \p\in\tra(i)$,
    
    \item $(\setL, \asg,\ctx)\models\lnot\phi$ if $(\setL, \asg,\ctx)\not\models\phi$,
    
    \item $(\setL, \asg,\ctx)\models\phi_1\lor\phi_2$ if $(\setL, \asg,\ctx)\models\phi_1$ or $(\setL, \asg,\ctx)\models\phi_2$,
    
    \item $(\setL, \asg,\ctx)\models\C{\ctx'}\phi$ if $(\setL, \asg,\ctx')\models\phi$,
    
    \item $(\setL, \asg,\ctx)\models\X_\setS\phi$ if $(\setL, \succ_{(\setS,\ctx)}(\asg),\ctx)\models\phi$,
    
    \item $(\setL, \asg,\ctx)\models\phi_1\U_\setS\phi_2$  if there exists an $i\geq 0$ such that $(\setL, \succ^i_{(\setS,\ctx)}(\asg),\ctx)\models \phi_2$ and
    
    $(\setL, \succ^j_{(\setS,\ctx)}(\asg),\ctx)\models\phi_1$ for all $0\leq j<i$,

    \item $(\setL, \asg,\ctx)\models\Y_\setS\phi$ if $\pred_{(\setS,\ctx)}(\asg)$ is defined and $(\setL, \pred_{(\setS,\ctx)}(\asg),\ctx)\models\phi$,
    
    \item $(\setL, \asg,\ctx)\models\phi_1\Since_\setS\phi_2 $ if there exists an $i\geq 0$ such that $\pred^i_{(\setS,\ctx)}(\asg)$ is defined,\newline $(\setL, \pred^i_{(\setS,\ctx)}(\asg),\ctx)\models \phi_2$, and $(\setL, \pred^j_{(\setS,\ctx)}(\asg),\ctx)\models \phi_1$ for all $0\leq j<i$,

    \item $(\setL, \asg,\ctx)\models\exists\tr.\phi$ if there exists a trace~$\tra\in\setL$ such that $(\setL, \asg[\tr\mapsto(\tra,0)],\ctx)\models\phi$, and 
    
    \item $(\setL, \asg,\ctx)\models\forall\tr.\phi$ if for all traces~$\tra\in\setL$ we have $(\setL, \asg[\tr\mapsto(\tra,0)],\ctx)\models\phi$.

\end{itemize}
Note that quantification ranges over \emph{initial} pointed traces, even under the scope of a temporal~operator.

We say that a set~$\setL$ of traces satisfies a sentence~$\phi$, written~$\setL \models \phi$, if $(\setL, \emptyset,\vars)\models\phi$, where $\emptyset$ represents the variable assignment with empty domain.
Furthermore, a transition system~$\TS$ satisfies $\phi$, written $\TS \models \phi$, if $\traces(\TS) \models \phi$.

\begin{remark}
\label{remarkquantfreesemantics}
Let $\asg$ be an assignment, $\ctx$ a context, and $\phi$ a quantifier-free \ghltl formula. Then, we have $(\setL, \asg, \ctx) \models \phi$ if and only if $(\setL', \asg, \ctx) \models \phi$ for all sets~$\setL,\setL'$ of traces, i.e., satisfaction of quantifier-free formulas is independent of the set of traces, only the assignment~$\asg$ and the context~$\ctx$ matter.
Hence, we will often write $(\asg,\ctx) \models\phi$ for quantifier-free $\phi$.
\end{remark}

\hyltl~\cite{ClarksonFKMRS14}, \hltls~\cite{DBLP:conf/lics/BozzelliPS21} (HyperLTL with stuttering), and \hltlc~\cite{DBLP:conf/lics/BozzelliPS21} (HyperLTL with contexts) are syntactic fragments of \ghltl.
Let us say that a formula is past-free, if it does not use the temporal operators~$\Y$ and $\Since$, Then,
\begin{itemize}
    \item \hyltl is the fragment obtained by considering only past-free \ghltl formulas in prenex normal form, by disallowing the context operator~$\C\cdot$, and by indexing all temporal operators by the empty set,
    \item \hltls is the fragment obtained by considering only past-free \ghltl formulas in prenex normal form, by disallowing the context operator~$\C\cdot$, and by indexing all temporal operators by sets of past-free \pltl formulas, and 
    \item \hltlc is the fragment obtained by considering only past-free \ghltl formulas in prenex normal form and by indexing all temporal operators by the empty set.
\end{itemize}

\textbf{Arithmetic and Complexity Classes for Undecidable Problems.} 
To capture the complexity of undecidable problems, we consider formulas of arithmetic, i.e., predicate logic with signature~$(+, \cdot, <, \in)$, evaluated over the structure~$\natsstruct$. 
A type~$0$ object is a natural number in $\nats$, and a type~$1$ object is a subset of $\nats$.
In the following, we use lower-case roman letters (possibly with decorations) for first-order variables, and upper-case roman letters (possibly with decorations) for second-order variables.
Every fixed natural number is definable in first-order arithmetic, so we freely use them as syntactic sugar. For more detailed definitions, we refer to \cite{Rogers87}.

Our benchmark is second-order arithmetic, i.e., predicate logic with quantification over type~$0$ and type~$1$ objects. 
Arithmetic formulas with a single free first-order variable define sets of natural numbers. In particular, $\Sigma_1^1$ contains the sets of the form~$\set{x \in\nats \mid \exists X_1 \subseteq \nats.\ \cdots \exists X_k\subseteq \nats.\ \psi(x, X_1, \ldots,X_k )}$, where $\psi$ is a formula of arithmetic with arbitrary quantification over type~$0$ objects (but no second-order quantifiers). 
Furthermore, truth in second-order arithmetic is the following problem: Given a sentence~$\phi$ of second-order arithmetic, do we have $\natsstruct\models\phi$?

\section{``Small'' Models for \texorpdfstring{\ghltl}{Generalized HyperLTL with Stuttering and Contexts}}

In this section, we prove that every satisfiable \ghltl sentence has a countable model, which is an important stepping stone for determining the complexity of the satisfiability problem in Section~\ref{secsat}. 
To do so, we first prove that for every \ghltl sentence~$\phi$ there is a $\ghltl$ sentence $\phi_p$ in prenex normal form that is \myquot{almost} equivalent in the following sense:
A set~$\setL$ of traces is a model of $\phi$ if and only if $\setL \cup \setL_\intdeco$ is a model of $\phi_p$, where $\setL_\intdeco$ is a countable set of traces that is independent of $\phi$.

Before we formally state our result, let us illustrate the obstacle we have to overcome, which traces are in $\setL_\intdeco$, and how they help to overcome the obstacle.
For the sake of simplicity, we use an always formula as example, even though it is syntactic sugar: The same obstacle occurs for the until operator, but there we would have to deal with the two subformulas of $\psi_1 \U_\setS \psi_2$ instead of the single one of $\G_\setS \psi$.

In a formula of the form~$\exists \tr. \G_\setS \exists \tr'.\ \psi$, the always operator acts like a quantifier too, i.e., the formula expresses that there is a trace~$\tra$ such that for \emph{every} position~$i$ on $\tra$, there is another trace~$\tra'$ (that may depend on $i$) so that $([\tr\mapsto (\tra,i), \tr'\mapsto(\tra',0)],\ctx)$ satisfies $\psi$, where $\ctx$ is the current context.
Obviously, moving the quantification of $\tr'$ before the always operator does not yield an equivalent formula, as $\tr'$ then no longer depends on $i$.
Instead, we simulate the implicit quantification over positions~$i$ by explicit quantification over natural numbers encoded by traces in~$\emptyset^i\set{\intprop}\emptyset^\omega$, where $\intprop \notin \ap$ is a fresh proposition.

Recall that $\setS$ is a set of \pltl formulas over $\ap$, i.e., Remark~\ref{remstutteringspecialcases} applies. Thus, the $i$-th $\setS$-successor of $(\emptyset^i\set{\intprop}\emptyset^\omega,0)$ is the unique pointed trace~$(\emptyset^i\set{\intprop}\emptyset^\omega, j)$ satisfying the formula~$\intprop$, which is the case for $j = i$. Thus, we can simulate the evaluation of the formula~$\psi$ at the $i$-th $(\setS,\ctx)$-successor by the formula~$\C{(\ctx \cup \{\tr_i\})\setminus\set{\tr'}}\F_\setS(\intprop_{\tr_i} \wedge \C \ctx \psi)$, where $\ctx$ is still the current context, i.e., we add $\tr_i$ to the current context to reach the $i$-th $\setS$-successor (over the extended context~$\ctx \cup \{\tr_i\}$) and then evaluate $\psi$ over the context~$\ctx$ that our original formula is evaluated over. But, to simulate the quantification of $\tr'$ correctly, we have to take it out of the scope for the eventually operator in order to ensure that the evaluation of $\psi$ takes place on the initial pointed trace, as we have moved the quantifier for $\tr'$ before the eventually.

To implement the same approach for the past operators, we also need to be able to let time proceed backwards from the position of $\emptyset^i\set{\intprop}\emptyset^\omega$ marked by $\intprop$ back to the initial position. 
To identify that position by a formula, we rely on the fact that $\neg \Y \true_{\tr}$ holds exactly at position~$0$ of the trace bound to $\tr$, where $\true_{\tr}$ is a shorthand for $p_{\tr} \vee \neg p_{\tr}$ for some $p \in\ap$.
Then, the $i$-th $\setS$-predecessor of the unique position marked by $\intprop$ is the unique position where $\neg \Y \true_{\tr}$ holds.
So, we define $\setL_\intdeco = \set{\emptyset^{i}\set{\intprop}\emptyset^\omega \mid i\in\nats}$. 

The proof of the next lemma shows that one can, in a similar way, move quantifiers over \emph{all} operators. 
\begin{lemma}\label{lem:pnn}
Let $\ap$ be a finite set of propositions and $\intprop\not\in\ap$.
    For every \ghltl sentence~$\phi$ over $\ap$, there exists a \ghltl sentence~$\phi_p$ in prenex normal form over $\ap\cup \{\intprop\}$ such that for all nonempty $\setL  \subseteq (\pow\ap)^\omega$:
    $\setL \models\phi$ if and only if $\setL \cup \setL_\intdeco\models\phi_p$.
\end{lemma}

The previous lemma shows that for every \ghltl sentence~$\phi$, there exists a \ghltl sentence~$\phi_p$ in prenex normal form that is almost equivalent, i.e., equivalent modulo adding the countable set~$\setL_\intdeco$ of traces to the model. 
Hence, showing that every satisfiable sentence in prenex normal form has a countable model implies that every satisfiable sentence has a countable model.
Thus, we focus on prenex normal form sentences to study the cardinality of models of \ghltl.

The proof of the next lemma shows that every model of a sentence~$\varphi$ contains a countable subset that is closed under the application of Skolem functions, generalizing a similar construction for \hyltl~\cite{FZ17}.
Such a set is also a model of $\varphi$.

\begin{lemma}\label{lem:countpnn}
    Every satisfiable \ghltl sentence~$\phi$ in prenex normal form has a countable model.
\end{lemma}

By combining Lemma~\ref{lem:pnn} and Lemma~\ref{lem:countpnn} we obtain our main result of this section.

\begin{theorem}\label{thm:countable}
    Every satisfiable \ghltl formula has a countable model.
\end{theorem}

\section{\texorpdfstring{\ghltl}{Generalized HyperLTL with Stuttering and Contexts} Satisfiability}
\label{secsat}

In this section, we study the satisfiability problem for \ghltl and its fragments: 
Given a sentence~$\phi$, is there a set~$\setL$ of traces such that $\setL \models \phi$?
Due to Theorem~\ref{thm:countable}, we can restrict ourselves to countable models.
The proof of the following theorem shows that the existence of such a model can be captured in arithmetic, generalizing constructions developed for \hyltl~\cite{hyperltlsat} to handle stuttering and contexts.

\begin{theorem}\label{thm:sat}
    The \ghltl satisfiability problem is $\Sigma^1_1$-complete.
\end{theorem}

As \hyltl is a fragment of \hltlc and of \hltls, which in turn are fragments of \ghltl, we also settle the complexity of their satisfiability problem as well. 

\begin{corollary}
  The \hltlc and \hltls satisfiability problems are both $\Sigma_1^1$-complete.
\end{corollary}

Thus, maybe slightly surprisingly, all four satisfiability problems have the same complexity, even though \ghltl adds stuttering, contexts, and quantification under the scope of temporal operators to \hyltl.
This result should also be compared with the \hyctlstar satisfiability problem, which is $\Sigma^2_1$-complete~\cite{hyperltlsat}, i.e., much harder.
\hyctlstar is obtained by extending \hyltl with just the ability to quantify under the scope of temporal operators. However, it has a branching-time semantics and trace quantification ranges over trace suffixes starting at the current position of the most recently quantified trace.
This allows one to write a formula that has only uncountable models, the crucial step towards obtaining the $\Sigma_1^2$-lower bound.
In comparison, \ghltl has a linear-time semantics and trace quantification ranges over initial traces, which is not sufficient to enforce uncountable models.

\section{Model-Checking \texorpdfstring{\ghltl}{Generalized HyperLTL with Stuttering and Contexts}}

In this section, we settle the complexity of the model-checking problems for \ghltl and its fragments: 
Given a sentence~$\phi$ and a transition system~$\TS$, do we have $\TS \models \phi$?
Recall that for \hyltl, model-checking is decidable. 
We show here that \ghltl model-checking is equivalent to truth in second-order arithmetic, with the lower bounds already holding for \hltlc and \hltls, i.e., adding only contexts and adding only stuttering makes \hyltl model-checking much harder.

The proof is split into three lemmata. 
We begin with the upper bound for full \ghltl.
Here, in comparison to the upper bound for satisfiability, we have to work with possibly uncountable models, as the transition system may have uncountably many traces.

The proof of the upper bound encodes the semantics of \ghltl over (possibly) uncountable sets of traces in arithmetic. 

\begin{lemma}\label{lem:mcG}
\ghltl model-checking is reducible to truth in second-order arithmetic.
\end{lemma}

Next, we prove the matching lower bounds for the two fragments~\hltls and \hltlc of \ghltl.
This shows that already stuttering alone and contexts alone reach the full complexity of \ghltl model-checking. 

We proceed as follows: traces over a single proposition~$\inprop$ encode sets of natural numbers, and thus also natural numbers (via singleton sets).
Hence, trace quantification in a transition system that has all traces over~$\set{\inprop}$ can mimic first- and second-order quantification in arithmetic.
The main missing piece is thus the implementation of addition and multiplication using only stuttering  and using only contexts. 
In the following, we present such implementations, thereby showing that one can embed second-order arithmetic in both \hltls and \hltlc.
To implement multiplication, we need to work with traces of the form~$\tra = \set{\auxprop}^{m_0} \emptyset^{m_1}\set{\auxprop}^{m_2} \emptyset^{m_3}\cdots$ for some auxiliary proposition~$\auxprop$.
We call a maximal infix of the form~$\set{\auxprop}^{m_j}$ or $\emptyset^{m_j}$ a block of $\tra$.
If we have $m_0 = m_1 = m_2 = \cdots$, then we say that $\tra$ is periodic and call $m_0$ the period of $\tra$.

\begin{lemma}\label{lem:mcS}
  Truth in second-order arithmetic is reducible to \hltls model checking.
\end{lemma}

\begin{proof}
We present a polynomial-time translation mapping sentences~$\phi$ of second-order arithmetic to pairs~$(\TS, \hyperize(\phi))$ of transition systems~$\TS$ and \hltls sentences~$\hyperize(\phi)$ such that $\natsstruct \models \phi$ if and only if $\TS \models \hyperize(\phi)$.
Intuitively, we capture the semantics of arithmetic in \hltls.

We begin by formalizing our encoding of natural numbers and sets of natural numbers using traces. 
Intuitively, a trace~$\tra$ over a set~$\ap$ of propositions containing the proposition~$\inprop$ encodes the set~$\set{n \in \nats \mid \inprop \in \tra(n)} \subseteq \nats$.
In particular, a trace~$\tra$ encodes a singleton set if it satisfies the formula~$(\neg \inprop) \U (\inprop \wedge \X\G\neg \inprop)$.
In the following, we use the encoding of singleton sets to encode natural numbers as well.
Obviously, every set and every natural number is encoded by a trace in that manner.
Thus, we can mimic first- and second-order quantification by quantification over traces.

However, to implement addition and multiplication using only stuttering, we need to adapt this simple encoding: 
We need a unique proposition for each first-order variable in $\phi$.
So, let us fix a sentence~$\phi$ of second-order arithmetic and let $V_1$ be the set of first-order variables appearing in $\phi$.
We use the set~$\ap = \set{\inprop} \cup \set{\inpropm{y} \mid y \in V_1} \cup \set{\auxprop, \auxpropp}$ of propositions, where $\auxprop$ and $\auxpropp$ are auxiliary propositions used to implement multiplication.

We define the function~$\hyperize$ mapping second-order formulas to \hltls formulas:
\begin{itemize}
    
    \item $\hyperize(\exists Y. \psi) = \exists \tr_Y. \left(\G_\emptyset \bigwedge_{p \neq \inprop} \neg p_{\tr_Y} \right) \wedge \hyperize(\psi)$.
    
    \item $\hyperize(\forall Y. \psi) = \forall \tr_Y. \left(\G_\emptyset \bigwedge_{p \neq \inprop} \neg p_{\tr_Y}\right) \rightarrow \hyperize(\psi)$.
    
    \item $\hyperize(\exists y. \psi) = \exists \tr_y. \left(\G_\emptyset \bigwedge_{p \neq \inpropm{y}} \neg p_{\tr_y}\right) \wedge \left((\neg (\inpropm{y})_{\tr_y}) \U_\emptyset ((\inpropm{y})_{\tr_y} \wedge \X_\emptyset\G_\emptyset\neg (\inpropm{y})_{\tr_y})\right) \wedge \hyperize(\psi)$.
    
    \item $\hyperize(\forall y. \psi) = \forall \tr_y. \left[\left(\G_\emptyset \bigwedge_{p \neq \inpropm{y}} \neg p_{\tr_y}\right) \wedge \left((\neg (\inpropm{y})_{\tr_y}) \U_\emptyset ((\inpropm{y})_{\tr_y} \wedge \X_\emptyset\G_\emptyset\neg (\inpropm{y})_{\tr_y})\right)\right] \rightarrow \hyperize(\psi)$.
    
    \item $\hyperize(\neg \psi) = \neg \hyperize(\psi)$.
    
    \item $\hyperize(\psi_1 \vee \psi_2) = \hyperize(\psi_1) \vee \hyperize(\psi_2)$.

    \item $\hyperize(y \in Y) = \F_\emptyset((\inpropm{y})_{\tr_y} \wedge \inprop_{\tr_Y})$.

    \item $\hyperize(y_1 < y_2) = \F_\emptyset( (\inpropm{y_1})_{\tr_{y_1}} \wedge \X_\emptyset\F_\emptyset  (\inpropm{y_2})_{\tr_{y_2}})$.
    
\end{itemize}
At this point, it remains to implement addition and multiplication in \hltls.
Let $\asg$ be an assignment that maps each $\tr_{y_j}$ for $j \in \set{1,2,3}$ to a pointed trace~$(\tra_j, 0)$ for some $\tra_j $ of the form
\[
\emptyset^{n_j}\set{\inpropm{y_j}}\emptyset^\omega \tag{$\dagger$}
\] which is the form of traces that trace variables~$\tr_{y_j}$ encoding first-order variables~$y_j$ of $\phi$ range over.
Our goal is to write formulas~$\hyperize(y_1 + y_2 = y_3)$ and $\hyperize(y_1 \cdot y_2 = y_3)$ with free variables~$\tr_{y_1}, \tr_{y_2}, \tr_{y_3}$ that are satisfied by $\asg$ if and only if $n_1 + n_2 = n_3$ and $n_1 \cdot n_2 = n_3$, respectively.

We begin with addition.
We assume that $n_1$ and $n_2$ are both non-zero and handle the special cases $n_1=0$ and $n_2=0$  later separately.
Consider the formula
\[
\alpha_\adddeco = \exists \traux. \left[\psi \wedge \X_{\inpropm{y_2}} \F_\emptyset ( (\inpropm{y_1})_{\tr_{y_1}} \wedge \X_\emptyset (\inpropm{y_3})_\traux )\right],
\]
where
$
\psi = 
\G_\emptyset [ (\inpropm{y_2})_{\tr_{y_2}} \leftrightarrow (\inpropm{y_2})_{\traux} ] \wedge 
\G_\emptyset [ (\inpropm{y_3})_{\tr_{y_3}} \leftrightarrow (\inpropm{y_3})_{\traux} ]
$
holds if the truth values of $\inprop{y_2}$ coincide on $\asg(y_2)$ and $\asg(\traux)$ and those of $\inprop{y_3}$ coincide on $\asg(y_3)$ and $\asg(\traux)$, respectively (see Figure~\ref{figstutteringadd}).

   \begin{figure}
   \centering
        \begin{tikzpicture}[thick]

        \node at (-.5,1) {$\tr_{y_1}$};
        \node at (-.5,0) {$\tr_{y_2}$};
        \node at (-.5,-1) {$\tr_{y_3}$};
        \node at (-.5,-2) {$\tr$};

        \draw[->, > = stealth] (0,1) -- (12,1);
        \draw[->, > = stealth] (0,0) -- (12,0);
        \draw[->, > = stealth] (0,-1) -- (12,-1);
        \draw[->, > = stealth] (0,-2) -- (12,-2);

        \foreach \i in {0,1,...,11}{
    \draw (\i,-.1) -- (\i, .1);
    \draw (\i,.9) -- (\i, 1.1);
    \draw (\i,-.9) -- (\i, -1.1);
    \draw (\i,-1.9) -- (\i, -2.1);
  }
\def\y{.4}
\node at (5, 1+\y) {\footnotesize$\inpropm{y_1}$}; 
\node at (4, 0+\y) {\footnotesize$\inpropm{y_2}$};      
\node at (9, -1+\y) {\footnotesize$\inpropm{y_3}$}; 
\node at (4, -2+\y) {\footnotesize$\inpropm{y_2}$};      
\node at (9, -2+\y) {\footnotesize$\inpropm{y_3}$}; 

\node[circle,fill=gray!20, minimum size =12,inner sep = 0] at (1,1+\y) {\footnotesize a};
\node[circle,fill=gray!20, minimum size =12,inner sep = 0] (a2) at (3.3,0+\y) {\footnotesize a};
\draw[->, >=stealth] (a2) edge[bend right=10] (4,0.15);
\node[circle,fill=gray!20, minimum size =12,inner sep = 0] at (1,-1+\y) {\footnotesize a};
\node[circle,fill=gray!20, minimum size =12,inner sep = 0] (a4) at (3.3,-2+\y) {\footnotesize a};
\draw[->, >=stealth] (a4) edge[bend right=10] (4,-1.85);

\node[circle,fill=gray!20, minimum size =12,inner sep = 0] (b1) at (4.3,1+\y) {\footnotesize b};
\draw[->, >=stealth] (b1) edge[bend right=10] (5,1.15);
\node[circle,fill=gray!20, minimum size =12,inner sep = 0] at (8,0+\y) {\footnotesize b};
\node[circle,fill=gray!20, minimum size =12,inner sep = 0] at (5,-1+\y) {\footnotesize b};
\node[circle,fill=gray!20, minimum size =12,inner sep = 0] at (8,-2+\y) {\footnotesize b};

\node[circle,fill=gray!20, minimum size =12,inner sep = 0] at (6,1+\y) {\footnotesize c};
\node[circle,fill=gray!20, minimum size =12,inner sep = 0] at (9,0+\y) {\footnotesize c};
\node[circle,fill=gray!20, minimum size =12,inner sep = 0] at (6,-1+\y) {\footnotesize c};
\node[circle,fill=gray!20, minimum size =12,inner sep = 0] (c4) at (9.7,-2+\y) {\footnotesize c};
\draw[->, >=stealth] (c4) edge[bend left=10] (9,-1.85);

\foreach \i in {0,1,...,11}{
\node at (\i,-2.5) {\footnotesize\i};
}
            
        \end{tikzpicture}
        \caption{The formula~$\alpha_\adddeco$ implements addition, illustrated here for $n_1 = 5$ and $n_2 = 4$. }
        \label{figstutteringadd}
        
    \end{figure}
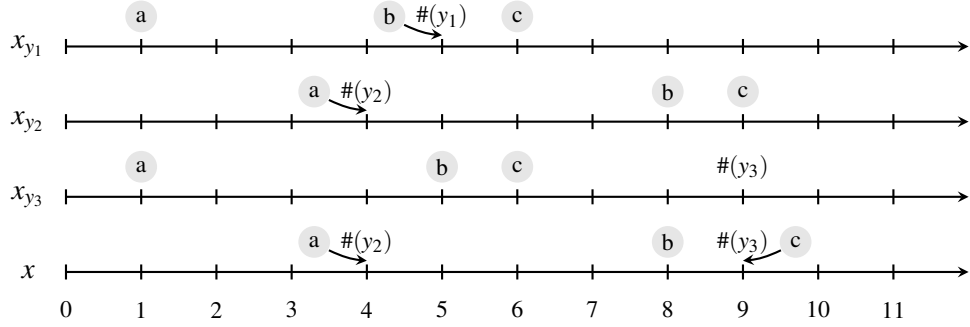

Now, consider an assignment~$\asg$ satisfying $\psi$ and ($\dagger$).
The outer next operator in $\alpha_\adddeco$ (labeled by $\inpropm{y_2}$) updates the pointers in $\asg$ to the ones marked by \myquot{a}.
For the traces assigned to $\tr_{y_1}$ and $\tr_{y_3}$ this is due to the fact that both traces do not contain~$\inpropm{y_2}$, which implies that every position is a $\inpropm{y_2}$-changepoint in these traces. 
On the other hand, both the traces assigned to $\tr_{y_2}$ and $\traux$ contain a $\inpropm{y_2}$. Hence, the pointers are updated to the first position where $\inpropm{y_2}$ holds (here, we use $n_2 > 0$).

Next, the eventually operator (labeled by $\emptyset$) updates the pointers in $\asg$ to the ones marked by \myquot{b}, as $\inpropm{y_1}$ has to hold on $\asg(\tr_{y_1})$. 
The distance between the positions marked \myquot{a} and \myquot{b} here is exactly $n_1 -1$ (here, we use $n_1 > 0$) and is applied to all pointers, as each position on each trace is a $\emptyset$-changepoint.

Due to the same argument, the inner next operator (labeled by $\emptyset$) updates the pointers in $\asg$ to the ones marked by \myquot{c}. 
In particular, on the trace~$\asg(\traux)$, this is position~$n_1 + n_2$.
As we require $\inpropm{y_3}$ to hold there (and thus also at the same position on $\asg(\tr_{y_3})$, due to $\psi$), we have indeed expressed $n_1 + n_2 = n_3$.

Accounting for the special cases $n_1 = 0$ (first line) and $n_2 = 0$ (second line), we define
\begin{align*}
\hyperize(y_1 + y_2 = y_3) = 
{}&{} \left[(\inpropm{y_1})_{\tr_{y_1}} \wedge \F_\emptyset( (\inpropm{y_2})_{\tr_{y_2}} \wedge (\inpropm{y_3})_{\tr_{y_3}} ) \right] \vee \\
{}&{}\left[ (\inpropm{y_2})_{\tr_{y_2}} \wedge \F_\emptyset( (\inpropm{y_1})_{\tr_{y_1}} \wedge (\inpropm{y_3})_{\tr_{y_3}} )\right] \vee \\
{}&{} \left[ \neg (\inpropm{y_1})_{\tr_{y_1}} \wedge \neg (\inpropm{y_2})_{\tr_{y_2}} \wedge \alpha_\adddeco \right].
\end{align*}

So, it remains to implement multiplication. 
Consider an assignment~$\asg$ satisfying $(\dagger)$. 
We again assume $n_1$ and $n_2$ to be non-zero and handle the special cases $n_1 =0$ and $n_2=0$ later. 
The formula
\begin{align*}
\alpha_1 = 
{}&{} \auxprop_\traux \wedge \G_\emptyset \F_\emptyset \auxprop_{\traux} \wedge \G_\emptyset \F_\emptyset \neg \auxprop_{\traux} \wedge \G_\emptyset \bigwedge\nolimits_{p \neq \auxprop,\inpropm{y_3}} \neg p_\traux \wedge \\
{}&{} \auxpropp_{\trauxp} \wedge \G_\emptyset \F_\emptyset \auxpropp_{\trauxp} \wedge \G_\emptyset \F_\emptyset \neg \auxpropp_{\trauxp} \wedge \G_\emptyset \bigwedge\nolimits_{p \neq \auxpropp} \neg p_{\trauxp} 
\end{align*}
with two (fresh) free variables~$\traux$ and $\trauxp$ expresses that 
    if $\asg(\traux) = (\tra, i)$ satisfies $i=0$, then $\tra$ is of the form~$\tra_{y_3} \cup \set{\auxprop}^{m_0}\emptyset^{m_1}\set{\auxprop}^{m_2}\emptyset^{m_3}\cdots$, for $\tra_{y_3} \in (\pow{\set{\inpropm{y_3}}})^\omega$ and non-zero~$m_j$, and
    if $\asg(\trauxp) = (\tra', i')$ satisfies $i'=0$, then $\tra'$ is of the form~$\set{\auxpropp}^{m_0'}\emptyset^{m_1'}\set{\auxpropp}^{m_2'}\emptyset^{m_3'}\cdots$, for non-zero~$m_j'$.
Here, $\cup$ denotes the pointwise union of two traces. 
The part~$\tr_{y_3}$ will only become relevant later, so we ignore it for the time being.

Then, the formula~$\alpha_2 = \G_\emptyset (\auxprop_{\traux} \leftrightarrow \auxpropp_{\trauxp}) $ is satisfied by $\asg$ if and only if $m_j = m_j'$ for all $j$.
If $\alpha_1\wedge\alpha_2$ is satisfied, then the $\set{\auxprop, \auxpropp}$-changepoints on $\tra$ \emph{and} $\tra'$ are $0, m_0, m_0 + m_1, m_0 + m_1 + m_2, \ldots$.
Now, consider 
\begin{align*}
\alpha_3 = \G_{\set{\auxprop, \auxpropp}}
\Big[ 
{}&{} \big[ 
\auxprop_\traux \rightarrow \X_{\auxpropp} \big( (\auxprop_\traux \wedge \neg\auxpropp_{\trauxp})\U_\emptyset (\neg \auxprop_\traux \wedge \neg\auxpropp_{\trauxp} \wedge \X_\emptyset \auxpropp_{\trauxp} ) \big) 
\big] \wedge\\
 {}&{}
\big[
\neg\auxprop_\traux \rightarrow \X_{\auxpropp} \big( (\neg\auxprop_\traux \wedge \auxpropp_{\trauxp})\U_\emptyset (\auxprop_\traux \wedge \auxpropp_{\trauxp} \wedge \X_\emptyset \neg\auxpropp_{\trauxp} ) \big) 
\big] 
\Big]   
\end{align*}
and a trace assignment~$\asg$ satisfying $\alpha_1 \wedge \alpha_2 \wedge \alpha_3$ and such that the pointers of $\asg(\traux)$ and $\asg(\trauxp)$ are both $0$.
Then, the always operator of $\alpha_3$ updates both pointers of $\asg(\traux)$ and $\asg(\trauxp)$ to some  $\set{\auxprop,\auxpropp}$-changepoint as argued above (see the positions marked \myquot{a} in Figure~\ref{figstutteringper}). This is the beginning of an infix of the form~$\set{\auxprop}^{m_j}$ in $\traux$ and of the form~$\set{\auxpropp}^{m_j}$ in $\trauxp$ or the beginning of an infix of the form~$\emptyset^{m_j}$ in $\traux$ and in $\trauxp$.

   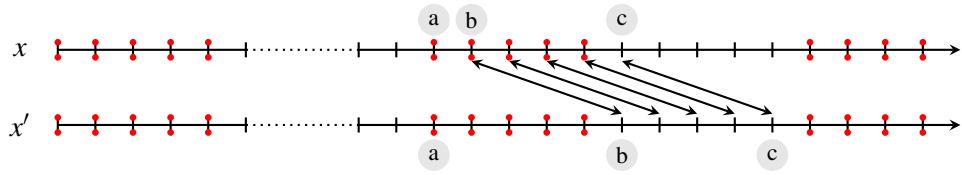
\begin{figure}
   \centering
        \begin{tikzpicture}[thick]

        \node at (-.5,1) {$\traux$};
        \node at (-.5,0) {$\trauxp$};

        \draw[] (0,1) -- (2.5,1);
        \draw[] (0,0) -- (2.5,0);
        \draw[dotted] (2.5,1) -- (4,1);
        \draw[dotted] (2.5,0) -- (4,0);
        \draw[->, > = stealth] (4,1) -- (12,1);
        \draw[->, > = stealth] (4,0) -- (12,0);

        \foreach \i in {0,.5,...,2.5,4,4.5,...,11.5}{
    \draw (\i,-.1) -- (\i, .1);
    \draw (\i,.9) -- (\i, 1.1);
  }

        \foreach \i in {0,.5,...,2,5,5.5,...,7,10,10.5,...,11.5}{
   \draw[fill,red] (\i,1.1) circle (.03);
\draw[fill,red] (\i,1-.1) circle (.03);

   \draw[fill,red] (\i,.1) circle (.03);
\draw[fill,red] (\i,-.1) circle (.03);
  }

\def\y{.4}
\node[circle, fill=gray!20, minimum size =12,inner sep = 0] at (5,1+\y) {\footnotesize a};
\node[circle, fill=gray!20, minimum size =12,inner sep = 0] at (5,0-\y) {\footnotesize a};
\node[circle, fill=gray!20, minimum size =12,inner sep = 0] at (5.5,1+\y) {\footnotesize b};
\node[circle, fill=gray!20, minimum size =12,inner sep = 0] at (7.5,0-\y) {\footnotesize b};
\node[circle, fill=gray!20, minimum size =12,inner sep = 0] at (7.5,1+\y) {\footnotesize c};
\node[circle, fill=gray!20, minimum size =12,inner sep = 0] at (9.5,0-\y) {\footnotesize c};

\foreach \i in {5.5,6,...,7.5}{
\draw[<->,>=stealth] (\i,.85) -- (\i+2,.15);
}
  




            
        \end{tikzpicture}
        \caption{The formula~$\alpha_3$ ensures that the traces assigned to $\traux$ and $\trauxp$ are periodic. Here, \myquot{\,\raisebox{-0.25ex}{\begin{tikzpicture}[thick]
\protect\draw(0,-.1) -- (0, .1);
\protect\draw[fill,red] (0,.1) circle (.03);
\protect\draw[fill,red] (0,-.1) circle (.03);
        \end{tikzpicture}}\,} denotes a position where $\auxprop$ ($\auxpropp$) holds and \myquot{\,\raisebox{-0.25ex}{\begin{tikzpicture}[thick]
\protect\draw(0,-.1) -- (0, .1);
        \end{tikzpicture}}\,} a position where $\auxprop$ ($\auxpropp$) does not hold.}
        \label{figstutteringper}
        
    \end{figure}

Let us assume we are in the former case, the latter is dual.
Then, the premise of the upper implication of $\alpha_3$ holds, i.e., $
\X_{\auxpropp} ( (\auxprop_\traux \wedge \neg\auxpropp_{\trauxp})\U_\emptyset (\neg \auxprop_\traux \wedge \neg\auxpropp_{\trauxp} \wedge \X_\emptyset \auxpropp_{\trauxp} ) )$ must be satisfied as well.
The next operator (labeled by $\set{\auxpropp}$ only) increments the pointer of $\asg(\traux)$ by one (as $\tra$ does not contain any $\auxpropp$) and updates the one of $\asg(\trauxp)$ to the position after the infix~$\set{\auxpropp}^{m_j}$ in $\tra'$ 
(see the pointers marked with \myquot{b}).
Now, the until formula only holds if we have $m_j = m_{j+1}$, as it compares the positions marked by the diagonal lines until it \myquot{reaches} the positions marked with \myquot{c}.
As the reasoning holds for every $j$ we conclude that we have $m_0 = m_1 = m_2 = \cdots$, i.e., we have constructed a periodic trace with period~$m_0$ that will allow us to implement multiplication by $m_0$.
From here on, we ignore $\asg(\trauxp)$, as it is only needed to construct~$\asg(\traux)$.

Next, we relate the trace~$\asg(\traux)$ to the traces~$\tr_{y_j}$ encoding the numbers we want to multiply using
\begin{align*}
\alpha_\multdeco = \exists \traux. \exists \trauxp.{}&{} \alpha_1 \wedge \alpha_2 \wedge \alpha_3 \wedge ( \auxprop_\traux \U_\emptyset (\neg \auxprop_\traux \wedge (\inpropm{y_1})_{\tr_{y_1}}) ) \wedge\\
{}&{}\G_\emptyset((\inpropm{y_3})_{\tr_{y_3}} \leftrightarrow (\inpropm{y_3})_{\traux}) \wedge \F_{\auxprop}[(\inpropm{y_2})_{\tr_{y_2}} \wedge (\inpropm{y_3})_{\traux} ],    
\end{align*}
which additionally expresses that $m_0$ is equal to $n_1$, that $\inpropm{y_3}$ holds on $\asg(\traux)$ at position~$n_3$ as well (and nowhere else), and finally that $n_3$ is equal to $n_1 \cdot n_2$:
This follows from the fact that we stutter $n_2$ times on $\tr_{y_2}$ to reach the position where $\inpropm{y_2}$ holds (as every position is a $\auxprop$-changepoint on $\tra_2$). Thus, we reach position~$m_0 \cdot n_2 = n_1 \cdot n_2$ on $\asg(\traux)$, at which $\inpropm{y_3}$ must hold. 
As $\inpropm{y_3}$ holds at the same position on $\tr_3$, we have indeed captured $n_1 \cdot n_2 = n_3$.

So, taking the special cases $n_1 = 0$ and $n_2 = 0$ (in the first two disjuncts) into account, we define
\begin{align*}
\hyperize(y_1 \cdot y_2 = y_3) = 
{}&{} [(\inpropm{y_1})_{\tr_{y_1}}\!\! \wedge (\inpropm{y_3})_{\tr_{y_3}}  ]\vee  [(\inpropm{y_2})_{\tr_{y_2}}\!\! \wedge (\inpropm{y_3})_{\tr_{y_3}} ] \vee 
  [\neg (\inpropm{y_1})_{\tr_{y_1}}\!\! \wedge \neg (\inpropm{y_2})_{\tr_{y_2}} \wedge \alpha_\multdeco].
\end{align*}

Now, let $\TS$ be a transition system with 
$
\traces(\TS) = (\pow{\set{\inprop}})^\omega \cup \bigcup_{y \in V_1}  (\pow{\set{\inpropm{y},\auxprop}})^\omega\cup (\pow{\set{\auxpropp}})^\omega$,
where $V_1$ still denotes the set of first-order variables in the sentence~$\phi$ of second-order arithmetic. Here, $(\pow{\set{\inprop}})^\omega$ contains the traces to mimic set quantification, 
$\bigcup_{y \in V_1}  (\pow{\set{\inpropm{y},\auxprop}})^\omega$ contains the traces to mimic first-order quantification and for the variable~$\traux$ used in the definition of multiplication, and $(\pow{\set{\auxpropp}})^\omega$ contains the traces for $\trauxp$, also used in the definition of multiplication.
Such a $\TS$ can, given $\phi$, be constructed in polynomial time.
An induction shows that we have $\natsstruct \models \varphi$ if and only if $\TS \models \hyperize(\phi)$. Note that while $\hyperize(\varphi)$ is not necessarily in prenex normal form, it can be brought into that as no quantifier is under the scope of a temporal operator, i.e., into a \hltls formula.
\end{proof}

Next, we consider the lower-bound for the second fragment, i.e., \hyltl with contexts.

\begin{lemma}\label{lem:mcC}
    Truth in second-order arithmetic is reducible to \hltlc model checking.
\end{lemma}

\begin{proof}
We present a polynomial-time translation mapping sentences~$\phi$ of second-order arithmetic to pairs~$(\TS, \hyperize(\phi))$ of transition systems~$\TS$ and \hltls sentences~$\hyperize(\phi)$ such that $\natsstruct \models \phi$ if and only if $\TS \models \hyperize(\phi)$.
Intuitively, we will capture the semantics of arithmetic in \hltlc.
As the temporal operators in \hltlc are all labeled by the empty set, we simplify our notation by dropping them in this proof, i.e., we just write $\X$, $\F$, $\G$, and $\U$.

As in the proof of Lemma~\ref{lem:mcS}, we encode natural numbers and sets of natural numbers by traces.
Here, it suffices to consider a single proposition~$\inprop$ to encode these and an additional proposition~$\auxprop$ that we use to implement multiplication, i.e., our \hltlc formulas are built over $\ap = \set{\inprop, \auxprop}$.

We again define a function~$\hyperize$ mapping second-order formulas to \hltlc formulas:
\begin{itemize}
    
    \item $\hyperize(\exists Y. \psi) = \exists \tr_Y. \left(\G \neg \auxprop_{\tr_Y} \right) \wedge \hyperize(\psi)$.
    
    \item $\hyperize(\forall Y. \psi) = \forall \tr_Y. \left(\G \neg \auxprop_{\tr_Y}\right) \rightarrow \hyperize(\psi)$.
    
    \item $\hyperize(\exists y. \psi) = \exists \tr_y. \left(\G \neg \auxprop_{\tr_y}\right) \wedge \left((\neg \inprop_{\tr_y}) \U (\inprop_{\tr_y} \wedge \X\G\neg \inprop_{\tr_y})\right) \wedge \hyperize(\psi)$.
    
    \item $\hyperize(\forall y. \psi) = \forall \tr_y. \left[\left(\G \neg \auxprop_{\tr_y}\right) \wedge \left((\neg \inprop_{\tr_y}) \U (\inprop_{\tr_y} \wedge \X\G\neg \inprop_{\tr_y})\right)\right] \rightarrow \hyperize(\psi)$.
    
    \item $\hyperize(\neg \psi) = \neg \hyperize(\psi)$.
    
    \item $\hyperize(\psi_1 \vee \psi_2) = \hyperize(\psi_1) \vee \hyperize(\psi_2)$.

    \item $\hyperize(y \in Y) = \F(\inprop_{\tr_y} \wedge \inprop_{\tr_Y})$.

    \item $\hyperize(y_1 < y_2) = \F( \inprop_{\tr_{y_1}} \wedge \X\F \inprop_{\tr_{y_2}})$.

\end{itemize}

At this point, it remains to consider addition and multiplication.
Let $\asg$ be an assignment that maps each $\tr_{y_j}$ for $j \in \set{1,2,3}$ to a pointed trace~$(\tra_j, 0)$ for some $\tra_j $ of the form~$ \emptyset^{n_j}\set{\inprop}\emptyset^\omega$, which is the form of traces that the $\tr_{y_j}$ encoding first-order variables~$y_j$ of $\phi$ range over.
Our goal is to write formulas~$\hyperize(y_1 + y_2 = y_3)$ and $\hyperize(y_1 \cdot y_2 = y_3)$ with free variables~$\tr_{y_1}, \tr_{y_2}, \tr_{y_3}$ that are satisfied by $(\asg,\vars)$ if and only if $n_1 + n_2 = n_3$ and $n_1 \cdot n_2 = n_3$, respectively.

The case of addition is readily implementable in \hltlc by defining
\[\hyperize(y_1 + y_2 = y_3) = \C{{\tr_{y_1},\tr_{y_3}}}\F\left[\inprop_{\tr_{y_1}}\land\C{{\tr_{y_2},\tr_{y_3}}}\F (\inprop_{\tr_{y_2}}\land \inprop_{\tr_{y_3}})\right].\]
The first eventually updates the pointers of $\asg(\tr_{y_1})$ and $\asg(\tr_{y_3})$ by adding $n_1$ and the second eventually updates the pointers of $\asg(\tr_{y_2})$ and $\asg(\tr_{y_3})$ by adding $n_2$. At that position, $\inprop$ must hold on $\tr_{y_3}$, which implies that we have $n_1 + n_2 = n_3$.

At this point, it remains to implement multiplication, which is more involved than addition.
In fact, we need to consider four different cases.
If $n_1 = 0$ or $n_2=0$, then we must have $n_3 = 0$ as well. This is captured by the formula~$\psi_1 = (\inprop_{\tr_{y_1}} \vee \inprop_{\tr_{y_2}}) \wedge \inprop_{\tr_{y_3}}$.
    Further, if $n_1 = n_2 = 1$, then we must have $n_3 = 1$ as well. This is captured by the formula~$
    \psi_2 = \X (\inprop_{\tr_{y_1}} \wedge \inprop_{\tr_{y_2}} \wedge \inprop_{\tr_{y_3}})$.

Next, let us consider the case~$0 < n_1 \le n_2$ with $n_2 \ge 2$.
    Let $z \in \nats\setminus\set{0}$ be minimal with 
    \begin{equation}
    \label{eq}
        z \cdot (n_2 - 1) = z' \cdot n_2 - n_1
    \end{equation}
    for some $z' \in \nats\setminus\set{0}$.
    It is easy to check that $z = n_1$ is a solution of Equation~(\ref{eq}) for $z' = n_1$.
    Now, consider some $0 < z < n_1$ to prove that $n_1$ is the minimal solution: 
    Rearranging Equation~\ref{eq} yields
    $    - z + n_1 = (z' - z)\cdot n_2
    $, 
    i.e., $-z + n_1$ must be a multiple of $n_2$ (possibly $0$).
    But $0 < z < n_1$ implies $0 < n_1 - z < n_1 \le n_2$, i.e., $-z + n$ is not a multiple of $n_2$. Hence, $z = n_1$ is indeed the smallest solution of Equation~(\ref{eq}).  
    So, for the minimal such $z$ we have $z \cdot (n_2 -1) + n_1 = n_1 \cdot n_2$, i.e., we have expressed multiplication of $n_1$ and $n_2$.

    Let us show how to implement this argument in \hltlc. 
    We begin by constructing two periodic traces with periods $n_2$ and $n_2-1$ using contexts.
    Consider the formula
    \[
    \alpha_1 = 
    (\auxprop_{\traux}  \wedge \G\F \auxprop_{\traux}  \wedge \G\F \neg \auxprop_{\traux} \wedge \G\neg\inprop_\traux ) \wedge 
    (\auxprop_{\trauxp} \wedge \G\F \auxprop_{\trauxp} \wedge \G\F \neg \auxprop_{\trauxp} \wedge \G\neg\inprop_{\trauxp} )
    \]
    with free variables~$\traux$ and $\trauxp$.
    If $\asg$ is an assignment satisfying $\alpha_1$ such that the pointers of $\asg(\traux)$ and $\asg(\trauxp)$ are both $0$, then $\asg(\traux) = (\set{\auxprop}^{m_0}\emptyset^{m_1}\set{\auxprop}^{m_2}\emptyset^{m_3}\cdots,0)$ and $\asg(\trauxp) =(\set{\auxprop}^{m_0'}\emptyset^{m_1'}\set{\auxprop}^{m_2'}\emptyset^{m_3'}\cdots,0)$ for some $m_j, m_j' \ge 1$.
    Further, the formula
    $
    \alpha_2 = \G (\auxprop_{\traux} \leftrightarrow \auxprop_{\trauxp})
    $
    then expresses that $m_j = m_j'$ for all $j$, and the formula
    \[
    \alpha_3 = \C{{\traux}}\left( \auxprop_{\traux}\U\left( \neg \auxprop_{\traux} \wedge \C{{\traux,\trauxp}}\G (\auxprop_{\traux} \leftrightarrow \neg\auxprop_{\trauxp}) \right)  \right)
    \]
    expresses that $m_j' = m_{j+1}$ for all $j$: The until operator updates the pointers of $\asg(\traux)$ and $\trauxp$ to the positions marked \myquot{a} in Figure~\ref{figcontextsperiodic} and the always operator then compares all following positions as depicted by the diagonal arrows.
    Thus, we have $m_j = m_0$ for all $j$ if $\asg(\traux)$ and $\asg(\trauxp)$ satisfy $\alpha_\per = \alpha_1 \wedge \alpha_2 \wedge \alpha_3$.

    \begin{figure}
    \centering
        \begin{tikzpicture}[thick]

        \node at (-.5,1) {$\traux$};
        \node at (-.5,0) {$\trauxp$};
        \def\y{.4}
        \node[fill=gray!20, circle, minimum size =12,inner sep = 0] at (2,1+\y) {a};
        \node[fill=gray!20, circle, minimum size =12,inner sep = 0] at (0,-\y) {a};
        
        \draw[->, > = stealth] (0,1) -- (12,1);
        \draw[->, > = stealth] (0,0) -- (12,0);

        \foreach \i in {0,1,...,11}{
    \draw (\i,-.1) -- (\i, .1);
    \draw (\i,.9) -- (\i, 1.1);
  }

        \foreach \i in {0,1,2,6,7,8}{
\draw[fill,red] (\i,.1) circle (.03);
\draw[fill,red] (\i,-.1) circle (.03);

\draw[fill,red] (\i,.9) circle (.03);
\draw[fill,red] (\i,1.1) circle (.03);
}

\begin{scope}
    \clip(0,0) rectangle (12,1);
  \foreach \i in {0,1,...,11}{
    \draw[<->,> = stealth,thin] (\i, .15) -- (\i+2,.85);
  }
\end{scope}

        \end{tikzpicture}
        \caption{The formula~$\alpha_3$ ensures that the traces assigned to $\traux$ (and $\trauxp$) are periodic. Here, \myquot{\,\raisebox{-0.25ex}{\begin{tikzpicture}[thick]
\protect\draw(0,-.1) -- (0, .1);
            \protect\draw[fill,red] (0,.1) circle (.03);
\protect\draw[fill,red] (0,-.1) circle (.03);
        \end{tikzpicture}}\,} (\myquot{\,\raisebox{-0.25ex}{\begin{tikzpicture}[thick]
\protect\draw(0,-.1) -- (0, .1);
        \end{tikzpicture}}\,}) denotes a position where $\auxprop$ holds (does not hold).}
        \label{figcontextsperiodic}
        
    \end{figure}
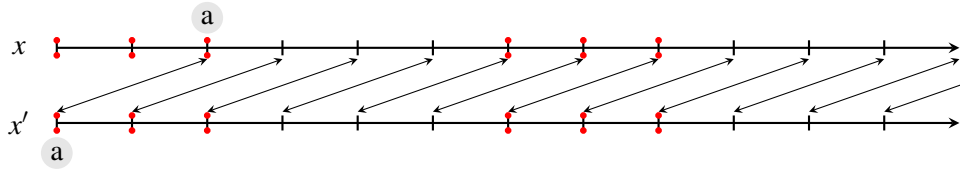

    To conclude, consider the formula
    \begin{align*}
    \psi_3 = {}&{}\Big[\X\F( \inprop_{\tr_{y_1}} \wedge \F \inprop_{\tr_{y_2}}) \wedge \X\X\F\inprop_{\tr_{y_2}}\Big] \rightarrow\\
    {}&{}\quad\Big[
    \exists \traux_0. \exists \trauxp_0.\exists \traux_1. \exists \trauxp_1. \alpha_\per[\traux/\traux_0, \trauxp/\trauxp_0] \wedge \alpha_\per[\traux/\traux_1, \trauxp/\trauxp_1] \wedge \\
    {}&{}\quad\left(\auxprop_{\traux_0}\U (\neg \auxprop_{\traux_0} \wedge \inprop_{\tr_{y_2}}) \right) \wedge \left(\auxprop_{\traux_1}\U ( \neg \auxprop_{\traux_1} \wedge \X \inprop_{\tr_{y_2}}) \right)\wedge\\
    {}&{}\quad\C{\tr_{y_1},\tr_{y_3},\traux_0}\left( \F\left( \inprop_{\tr_{y_1}} \wedge \C{\tr_{y_3},\traux_0, \traux_1} (\neg \alpha_\blockchange)\U (\alpha_\blockchange \wedge \X \inprop_{\tr_{y_3}}) \right) \right)
    \Big]
    \end{align*}
    where $\alpha_\blockchange = (\auxprop_{\traux_0} \leftrightarrow \neg\X\auxprop_{\traux_0}) \wedge (\auxprop_{\traux_1} \leftrightarrow \neg\X\auxprop_{\traux_1})  $ holds at $\asg(\traux_0)$ and $\asg(\traux_1)$ if the pointers both point to the end of a block on the respective trace.
    Furthermore, $\alpha_\per[\traux/\traux_j, \trauxp/\trauxp_j]$ denotes the formula obtained from replacing each occurrence of $\traux$ by $\traux_j$ and every occurrence of $\trauxp$ by $\trauxp_j$.
    Thus, we quantify two traces~$\asg(\traux_0) = (\tra_0,0)$ and $\asg(\traux_1)= (\tra_1,0)$ (we disregard $\asg(\trauxp_0)$ and $\asg(\trauxp_1)$, as we just need them to construct the former two traces) that are periodic. 

        \begin{figure}
        \centering
        \begin{tikzpicture}[thick]

    \def\y{1}
        \node at (-.5,\y) {$\tr_0$};
        \node at (-.5,0) {$\tr_1$};
        \node at (-.5,-\y) {$\tr_{y_3}$};

        \draw[->, > = stealth] (0,\y) -- (13,\y);
        
        \draw[->, > = stealth] (0,0) -- (13,0);

        \draw[->, > = stealth] (0,-\y) -> (13,-\y);
        
        \foreach \i in {0,.5,...,3,7,7.5,...,10}{
        \draw[fill,red] (\i,\y+.1) circle (.03);
        \draw[fill,red] (\i,\y+-.1) circle (.03);
        }

        \foreach \i in {0,.5,...,2.5,6,6.5,...,8.5, 12,12.5}{
        \draw[fill,red] (\i,.1) circle (.03);
        \draw[fill,red] (\i,-.1) circle (.03);
        }

        \foreach \i in {0,.5,...,12.5}{
    \draw (\i,-.1) -- (\i, .1);
    \draw (\i,\y-.1) -- (\i, \y+.1);
    \draw (\i,-\y-.1) -- (\i, -\y+.1);
    \pgfmathtruncatemacro{\result}{2*\i}
        \node at (\i,-1.25) {\scriptsize \result};
  }

      \foreach \i in {0,.5,...,8.5}{
    \draw[<->,> = stealth,thin] (\i, .15) -- (\i+1.5,\y-.15);
  }

        \node[fill=gray!20, circle, minimum size =12,inner sep = 0] at (1.5,\y+.4) {\footnotesize a};
        \node[fill=gray!20, circle, minimum size =12,inner sep = 0] at (0,.4) {\footnotesize a};
        \node[fill=gray!20, circle, minimum size =12,inner sep = 0] at (1.5,-\y+0.4) {\footnotesize a};
        
        \node[fill=gray!20, circle, minimum size =12,inner sep = 0] at (8.5,.4) {\footnotesize b};
        \node[fill=gray!20, circle, minimum size =12,inner sep = 0] at (10,\y+.4) {\footnotesize b};
        \node[fill=gray!20, circle, minimum size =12,inner sep = 0] at (10,-\y+.4) {\footnotesize b};
        
        \node[fill=gray!20, circle, minimum size =12,inner sep = 0] at (10.5,-\y+.4) {\footnotesize c};

        \end{tikzpicture}
        \caption{The formula~$\psi_3$ implementing multiplication, for $n_1 = 3$ and $n_2 = 7$, i.e., $\tr_0$ has period~$7$ and $\tr_1$ has period~$6$. Here, \myquot{\,\raisebox{-0.25ex}{\begin{tikzpicture}[thick]
\protect\draw(0,-.1) -- (0, .1);
            \protect\draw[fill,red] (0,.1) circle (.03);
\protect\draw[fill,red] (0,-.1) circle (.03);
        \end{tikzpicture}}\,} (\myquot{\,\raisebox{-0.25ex}{\begin{tikzpicture}[thick]
\protect\draw(0,-.1) -- (0, .1);
        \end{tikzpicture}}\,}) denotes a position where $\auxprop$ holds (does not hold).}
        \label{figcontextsmult}
        
    \end{figure}
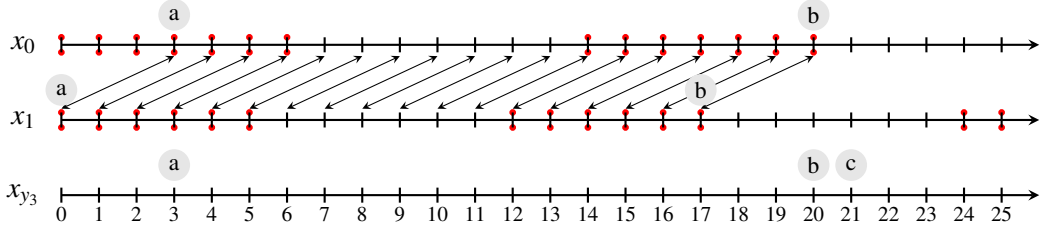
    
    The second line of $\psi_3$ is satisfied if $\tra_0$ has period~$n_2$ and $\tra_1$ has period $n_2-1$.
    Now, consider the last line of $\psi_3$ and see Figure~\ref{figcontextsmult}:
     %
        The first eventually-operator updates the pointers of $\asg(\tr_{y_1})$, $\asg(\tr_{y_3})$, and $\asg(\traux_0)$ to $n_1$, as this is the unique position on $\asg(\tr_{y_1})$ that satisfies~$\inprop$. Crucially, the pointer of $\asg(\traux_1)$ is not updated, as it is not in the current context. These positions are marked by \myquot{a}.

        Then, the until-operator compares positions~$i$ on $\tr_1$ and $i+n_1$ on $\tr_0$ (depicted by the diagonal lines) and thus subsequently updates the pointer of $\asg(\traux_1)$ to $x \cdot (n_2 -1) -1$ for the smallest $z \in\nats\setminus{0}$ such that $x \cdot (n_2 - 1) = z' \cdot n_2 - n_1$ (recall that the pointer of $\asg(\traux_0)$ with period~$n_2$ is already $n_1$ positions ahead) for some $z' \in \nats\setminus\set{0}$, as $\alpha_\blockchange$ only holds at the ends of the blocks of $\tr_0$ and $\tr_1$. 
        Accordingly, the until-operator updates the pointer of $\asg(\traux_0)$ to $z \cdot (n_2 -1) -1 + n_1$ and the pointer of $\asg(\tr_{y_3})$ to $z \cdot (n_2 -1) -1 + n_1$.
        These positions are marked by \myquot{b}.
        Then, the next-operator subsequently updates the pointer of $\asg(\tr_{y_3})$ to $z \cdot (n_2 -1) + n_1$, which is marked by \myquot{c} in Figure~\ref{figcontextsmult}.
    
    As argued above, $z$ must be equal $n_1$, i.e., the pointer of $\asg(\tr_{y_3})$ is then equal to $n_1 \cdot (n_2 -1) + n_1 = n_1 \cdot n_2$. At that position, $\psi_3$ requires that $\inprop$ holds on $\asg(\tr_{y_3})$. 
    Hence, we have indeed implemented multiplication of $n_1$ and $n_2$, provided we have $0 < n_1 \le n_2$ and $n_2 \ge 2$.

For the final case, i.e., $0 < n_2 < n_1$, we use a similar construction, but swap the roles of $y_1$ and $y_2$ in $\psi_3$, to obtain a formula~$\psi_4$.
Then, let $\hyperize(y_1 \cdot y_2 = y_3) = \psi_1 \vee\psi_2 \vee\psi_3 \vee\psi_4$ and let $\TS$ be a fixed transition system with 
$
\traces(\TS) = (\pow{\set{\inprop}})^\omega \cup (\pow{\set{\auxprop}})^\omega$,
Here, $(\pow{\set{\inprop}})^\omega$ contains the traces to mimic set quantification and $(\pow{\set{\auxprop}})^\omega$ contains the traces for $\traux_j$ and $\trauxp_j$ used to implement multiplication.
An induction shows that we have $\natsstruct \models \varphi$ if and only if $\TS \models \hyperize(\phi)$. 
Again, $\hyperize(\varphi)$ can be turned into a \hltlc formula (i.e., in prenex normal form), as no quantifier is under the scope of a temporal operator.
\end{proof}

Combining the results of this section, we obtain our main result settling the complexity of model-checking for \ghltl and its fragments \hltls and \hltlc.

\begin{theorem}
The model-checking problems for the logics~\ghltl, \hltls, and \hltlc are all equivalent to truth in second-order arithmetic.
\end{theorem}

\section{Conclusion}

In this work, we have settled the complexity of \ghltl, an expressive logic for the specification of asynchronous hyperproperties.
Although it is obtained by adding stuttering, contexts, and trace quantification under the scope of temporal operators to \hyltl, we have proven that its satisfiability problem is as hard as that of its (much weaker) fragment \hyltl. 
On the other hand, model-checking \ghltl is much harder than for \hyltl, i.e., equivalent to truth in second-order arithmetic vs.\ decidable. 
Here, the lower bounds again hold for simpler fragments, i.e., \hltls and \hltlc. 

Our work extends a line of work that has settled the complexity of synchronous hyperlogics like \hyltl~\cite{hyperltlsat}, \hyqptl~\cite{hyq}, and second-order \hyltl~\cite{fz25}.
In future work, we aim to resolve the exact complexity of other logics for asynchronous hyperproperties proposed in the literature. 